\documentclass[11pt]{article}
\pdfoutput=1
\usepackage[margin=1in]{geometry}
\usepackage{amsmath,amssymb}
\usepackage{comment}
\usepackage{tikz}
  \usetikzlibrary{automata,positioning,matrix,calc,petri,backgrounds}

\usepackage{caption}
\usepackage{subfigure}
         \usepackage{floatrow}
         
     \usepackage{hyperref}

\newcommand{\BSWAPInf}{{\mathcal A}_\psyst^\infty}

\newtheorem{theorem}{Theorem}
\newtheorem{lemma}{Lemma}

\newtheorem{proof}{Proof}
\newcommand*{\qed}{\hfill\ensuremath{\square}}%

\newcommand\tpl[1]{\langle{#1}\rangle}

\newcommand{\tup}[1]{\langle #1 \rangle}

\newcommand{\src}{\textsf{src}}
\newcommand{\dst}{\textsf{dst}}
\newcommand{\actn}{\textsf{actn}}

\newcommand{\brd}{\mathfrak{b}}
\def\Actions{\Sigma_{\textsf{actn}}}
\def\Actionprts{\Sigma_{\textsf{rdz}}}
\def\Comm{\Sigma_{\textsf{com}}}

\newcommand\restr[2]{{
  \left.\kern-\nulldelimiterspace 
  #1 
  \vphantom{\big|} 
  \right|_{#2} 
  }}

\newcommand{\proctemp}{P}
\newcommand{\uwd}{\multimap}
\newcommand{\rwd}{\circledcirc}
\newcommand{\procuwd}{P^\multimap}
\newcommand{\sysinst}[1]{{\mathcal{P}^{#1}}}

\newcommand{\psyst}{{\mathcal P}}
\newcommand{\psysttimed}{{\mathcal T}}

\newcommand{\runs}{\text{runs}}
\newcommand{\comp}{\text{comp}}
\newcommand{\edge}{\text{edge}}

\newcommand{\procs}{\text{prcs}}

\newcommand{\proj}[1]{\text{proj}_{#1}}
\newcommand{\exec}{\textsc{exec}_\psyst}
\newcommand{\execLimited}{\textsc{exec}_{\psysttimed}}

\newcommand{\execfin}[1][\psyst]{\textsc{exec}_{#1}^{\text fin}}
\newcommand{\execinf}[1][\psyst]{\textsc{exec}_{#1}^\infty}

\newcommand{\execNFW}{{\mathcal A}_\psyst^{fin}}
\newcommand{\execBSW}{{\mathcal A}_\psyst^{\infty}}

\newcommand{\A}{{\mathcal A}}

\newcommand{\bgood}{\textsf{blue}}
\newcommand{\good}{\textsf{green}}
\newcommand{\sgood}{\textsf{orange}}
\newcommand{\bad}{\textsf{red}}

\newcommand{\Pspec}{{\mathcal F}}

\newcommand{\PSPACE}{\textsc{pspace}}
\newcommand{\EXPSPACE}{\textsc{expspace}}

\newcommand\msg[1]{\ensuremath{\mathsf{#1}}}

\def \st {\, \vert\, }

\newcommand\head[1]{\smallskip\noindent\textbf{#1.}}

\newcommand\LTL{{\textsf{LTL}}}

\newcommand\Nat{\mathbb{N}}

\newcommand\Rat{\mathbb{Q}}

\newcommand\RatGEZ{\mathbb{Q}_{\geq 0}}

\def\NatZero{\mathbb{N}_0}

\newcommand{\timedStates}{\ensuremath{Q}}
\newcommand{\clocks}{\ensuremath{C}}

\date{}

\begin{document}

\title{Liveness of Parameterized Timed Networks
 \thanks{Benjamin Aminof and Florian Zuleger were supported by the Austrian National Research Network S11403-N23 (RiSE) of the Austrian Science Fund (FWF) and by the Vienna Science and Technology Fund (WWTF) through grant ICT12-059.
 Sasha Rubin is a Marie Curie fellow of the Istituto Nazionale di Alta
 Matematica. The final publication is available at Springer via \url{http://dx.doi.org/10.1007/978-3-662-47666-6_30}
 }
}
\author{
Benjamin Aminof\\
Technische Universit\"at Wien, Austria
\and 
Sasha Rubin\\
Universit\`a degli Studi di Napoli ``Federico II", Italy
\and
Francesco Spegni\\
Universit\`a Politecnica delle Marche, Ancona, Italy
\and
Florian Zuleger\\
Technische Universit\"at Wien, Austria}

\maketitle

\begin{abstract}
We consider the model checking problem of infinite state systems given in the form of parameterized discrete timed networks with multiple clocks.
We show that this problem is decidable with respect to specifications given by B- or S-automata.
Such specifications are very expressive (they strictly subsume $\omega$-regular specifications), and easily express complex liveness and safety properties.
Our results are obtained by modeling the passage of time using symmetric broadcast, and by solving the model checking problem of parameterized systems of untimed processes communicating using $k$-wise rendezvous and symmetric broadcast.
Our decidability proof makes use of automata theory, rational linear programming, and geometric reasoning for solving certain reachability questions in vector addition systems; we believe these proof techniques will be useful in solving related problems.
\end{abstract}

\section{Introduction}
Timed automata --- finite state automata enriched by a finite number of dense- or discrete-valued clocks --- can be used to model more realistic circuits and protocols than untimed systems~\cite{alur99,chevallier09}.  A timed network consists of an arbitrary but fixed number of timed automata running in parallel~\cite{aj03,ADM04}. In each computation step, either some fixed number of automata synchronize by a rendezvous-transition or time advances. We consider the \emph{parameterized model-checking problem (PMCP)} for timed networks: Does a given specification (usually given by a suitable automaton) hold \emph{for every system size}? Apart from a single result which deals with much weaker synchronization than rendezvous~\cite{ss14}, no positive PMCP results for liveness specifications of timed automata are known. 

\paragraph{System model:}
In this paper we prove the decidability of the PMCP for discrete timed networks with no controller and liveness specifications. To do this, we reduce the PMCP of these timed networks to the PMCP of {\em RB-systems} --- systems of finite automata communicating via \emph{$k$-wise rendezvous} and \emph{symmetric broadcast}.
This broadcast action is symmetric in the sense that there is no designated sender.
In contrast, the standard broadcast action can distinguish between sender and receivers, and so the PMCP of liveness properties is undecidable even in the untimed setting~\cite{efm99}.

\paragraph{Our Techniques and Results:}
Classical automata (e.g., nondeterministic B\"uchi word automata (NBW)) are not able to capture the behaviors of RB-systems.  Thus, our decidability result uses nondeterministic BS-automata (and their fragments B- and S-automata) which strictly subsume NBW~\cite{b10}.

We show that the PMCP is decidable for controllerless discrete timed networks and (and systems communicating via $k$-wise rendezvous and symmetric broadcast) and specifications given by B-automata or S-automata (and in particular by NBW) or for negative specifications (i.e., the set of bad executions) given by BS-automata. We prove decidability by constructing a B-automaton that precisely characterizes the runs of a timed network from the point of view of a single process.
Along the way, we also obtain an {\sc ExpSpace} upper bound for the PMCP of safety properties of discrete timed networks.

In order to build the B-automaton, an intricate analysis of the interaction between the transitions caused by the passage of time  (modeled by broadcasts) which involve all processes, and those that are the result of rendezvousing processes, is needed.
It is this interaction that makes the problem complicated. Thus, for example, results concerning pairwise rendezvous without broadcast~\cite{gs92} do not extend to our case.  Our solution to this problem involves the introduction of the idea of a {\em rational relaxation} of a Vector Addition System, and geometric lemmas concerning paths in these relaxations. It is important to note that these vector addition systems can not capture the edges that correspond to the passage of time. However, they provide the much needed flexibility in capturing what happens in between time ticks {\em in the presence of} these ticks.

\paragraph{Related Work.}

Discrete timed networks with rendezvous and a controller were introduced in \cite{ADM04} where it was shown that safety is decidable using the technique of well-structured transition systems. Their result implies a non-elementary upper bound (which we improve to \textsc{ExpSpace}) for the complexity of the PMCP of safety properties of timed networks without a controller. PMCP of liveness properties for continuous-time networks with a controller process is undecidable~\cite{aj03}. However, their proof heavily relies on time being dense and on the availability of a distinguished controller process.
RB-systems with a controller were introduced in \cite{KouvarosL13} where it is proved that under an additional strong restriction on the environment and process templates (called a shared simulation), such systems admit cutoffs that allow one to model check epistemic-temporal logic of the parameterised systems. The main difference between our work and theirs is: we do not have a controller, we make no additional restrictions, and we can model check specifications given by B- or S-automata. The authors in \cite{ss14} proved that the PMCP is decidable for continuous timed networks synchronizing using conjunctive Boolean guards and MITL and TCTL specifications. Finally, there are many decidability and undecidability results in the untimed setting,  e.g.,\cite{Su88,efm99,DelzannoSZ10,AJKR14,concur14}.

\section{Definitions and Preliminaries}\label{sec: definitions}

\head{Labeled Transition Systems}
A \emph{(edge-)labeled transition system (LTS)} is a tuple
$\tup{S,I,R,\Sigma}$, where $S$ is the set of {\em states} (usually $S \subseteq \Nat$),
$I \subseteq S$ are the {\em initial states}, $R \subseteq S \times \Sigma \times S$ is the {\em edge relation},
and $\Sigma$ is the {\em edge-labels alphabet}. {\em Paths} are sequences of transitions, and runs are paths starting in initial states.

\head{Automata}
We use standard notation and results of automata, such as nondeterministic B\"uchi word automata (NBW) \cite{vardi96}. A {\em BS-word automaton (BSW)} (\cite{b10}) is a tuple $\tup{\Sigma, Q, Q_0, \Gamma, \delta, \Phi}$ where $\Sigma$ is a finite {\em input alphabet}, $Q$ is a set of {\em states}, $Q_0 \subseteq Q$ is a set of {\em initial states}, $\Gamma$ is a set of {\em counter (names)}, $\delta \subseteq Q \times \Sigma \times \mathcal{C}^* \times Q$ is the {\em transition relation} where $\mathcal{C}$ is the set of {\em counter operations}, i.e. $c := 0, c : = c + 1, c := d$ for $c,d \in \Gamma$, and $\Phi$ is the {\em acceptance condition} described below. A run $\rho$ is defined like for nondeterministic automata over infinite words by ignoring the $\mathcal{C}^*$ component. Denote by $c(\rho,i)$ the $i$th value assumed by counter $c \in \Gamma$ along $\rho$. The acceptance condition $\Phi$ is a positive Boolean combination of the following conditions ($q \in Q, c \in \Gamma$): (i) $q$ is visited infinitely often (B\"uchi-condition); (ii) $\limsup_i c(\rho,i) < \infty$ (B-condition); (iii)  $\liminf_i c(\rho,i) = \infty$ (S-condition).  An automaton that does not use B-conditions is called an {\em S-automaton (SW)}, and one that does not use S-conditions is called a {\em B-automaton (BW)}.

It is known that BSWs are relatively well behaved~\cite{b10}: their emptiness problem is decidable; they are closed under union and intersection, but not complement; and BW (resp. SW) can be complemented to SW (resp. BW).
Since BSWs are not closed under complement, we are forced, if we are to use the automata-theoretic approach for model checking (cf. \cite{vardi96}), to give the specification in terms of the undesired behaviours, or to consider specifications in terms of BWs or SWs (which both strictly extend $\omega$-regular languages).

\head{Rendezvous with Symmetric Broadcast (RB-System)}
Intuitively, RB-systems describe the parallel composition of $n \in \Nat$ copies of a process {\em template}.
An RB-system evolves nondeterministically: either a $k$-wise rendezvous action is taken, i.e., $k$ different processes instantaneously synchronize on a rendezvous action $\msg{a}$, or the symmetric broadcast action is taken, i.e., all processes must take an edge labeled by $\brd$. 
Systems without the broadcast action are called R-systems.

In the rest of the paper, fix $k$ (the number of processes participating in a rendezvous), a finite set $\Actions$ of {\em rendezvous actions}, the \emph{rendezvous alphabet} $\Actionprts = \cup_{\msg{a} \in \Actions} \{ \msg{a}_1, \dots, \msg{a}_k\}$, and the  \emph{communication alphabet} $\Comm$ which is the union $\{((i_1,\msg{a}_1),\ldots,(i_k,\msg{a}_k)) \mid \msg{a} \in \Actions, i_j \in \mathbb{N}, j \in [k] \}\cup \{\brd\}$.

A \emph{process template} (or {\em RB-template}) is a finite LTS $\proctemp = \tup{S,I,R,\Actionprts \cup \{\brd\}}$ such that for every state $s \in S$ there is a transition $(s,\brd,s') \in R$ for some $s' \in S$. We call edges labeled by $\brd$ {\em broadcast edges}, and the rest {\em rendezvous edges}. For ease of exposition, we assume (with one notable exception, namely $\procuwd$ defined in Section~\ref{sec: unwinding}) that for every $\varsigma \in \Actionprts$ there is at most one edge in $\proctemp$ labeled by $\varsigma$ and we denote it by $\edge(\varsigma)$.\footnote{This can always be assumed by increasing the size of the rendezvous alphabet.}
.
The {\em RB-system $\sysinst{n}$} is defined, given a template $\proctemp$ and $n \in \Nat$, is defined as the finite LTS $\tup{Q^n,Q^n_0,\Delta^n,\Comm}$\footnote{Even though $\Comm$ is infinite, $\Delta^n$ refers only to a finite subset of it.}
where:
\begin{enumerate}
\item $Q^n$ is the set of functions (called {\em configurations}) of the form $f:[n] \to S$.
    We call $f(i)$ the {\em state of process $i$} in $f$. Note that we sometimes find it convenient to consider a more flexible naming of processes in which we let $Q^n$ be the set of functions $f:X \to S$, where $X \subset \Nat$ is some set of size $n$.

\item The set of {\em initial configurations} $Q^n_0 = \{ f \in Q^n \mid f(i) \in I \text{ for all } i \in [n]  \}$ consists of all configurations which map all processes to initial states of $\proctemp$.

\item The set of {\em global transitions} $\Delta^n$ are tuples $(f,\sigma,g) \in Q^n \times \Comm \times Q^n$ where one of the following two conditions hold:
    \begin{itemize}
    \item $\sigma = \brd$, and for every $i \in [n]$ we have that $(f(i), \brd, g(i)) \in R$.
        This is called a {\em broadcast transition}.
		
    \item $\sigma = ((i_1, \msg{a}_1), \dots, (i_k, \msg{a}_k))$, where $\msg{a} \in \Actions$ is the {\em action} taken, and $\{i_1, \dots, i_k\} \subseteq [n]$ are $k$ different processes.
        In this case, for every $1 \leq j \leq k$ we have that $(f(i_j),\msg{a}_j,g(i_j)) \in R$; and $f(i) = g(i)$ for every $i \not \in \{i_1, \dots, i_k\}$.
        This is called a {\em rendezvous transition}, and the processes in the set $\procs(\sigma) := \{i_1, \dots, i_k\}$ are called the {\em rendezvousing processes}.
	\end{itemize}
\end{enumerate}

We denote the {\em action taken} on a global transition $t = (f,\sigma,g)$ by $\actn(t)$. Thus, $\actn(t) := \msg{a}$ if $\sigma = ((i_1, \msg{a}_1), \dots, (i_k, \msg{a}_k))$, and otherwise $\actn(t) := \brd$.

A process template $\proctemp$ induces the {\em infinite RB-system} $\psyst$, i.e., the LTS $\psyst = \tup{Q,Q_0,\Delta,\Comm}$ where $Q = \cup_{n \in \Nat} Q^n$, $Q_0 = \cup_{n \in \Nat} Q^n_0$, $\Delta = \cup_{n \in \Nat} \Delta^n$.

\head{Executions of an RB-System, and the Parameterized Model-Checking Problem}
Given a global transition $t = (f,\sigma,g)$, and a process $i$, we say that $i$ {\em moved} in $t$ iff: $\sigma = \brd$, or $i \in \procs(\sigma)$.
We write $edge_i(t)$ for the edge of $\proctemp$ taken by process $i$ in the transition $t$, and $\bot$ if $i$ did not move in $t$.
Thus, if $\sigma = \brd$ then $edge_i(t) := (f(i), \brd, g(i))$; and if $\sigma = ((i_1, \msg{a}_1), \dots, (i_k, \msg{a}_k))$ then $edge_i(t) := (f(i), \msg{a}_j, g(i))$ if $\sigma(j) = (i, \msg{a}_j)$ for some $j \in [k]$, and otherwise $edge_i(t) := \bot$.
Take an RB-System $\sysinst{n} = \tup{Q^n,Q^n_0,\Delta^n,\Comm}$, a path $\pi = t_1 t_2 \dots$ in $\sysinst{n}$, and a process $i$ in $\sysinst{n}$. Define $\proj{\pi}(i) := edge_i(t_{j_1}) edge_i(t_{j_2}) \dots$, where $j_1 < j_2 < \dots$ are all the indices $j$ for which $edge_i(t_j) \neq \bot$. Intuitively, $\proj{\pi}(i)$ is the path in $\proctemp$ taken by process $i$ during the path $\pi$. Define the set of {\em executions} $\exec$ of $\psyst$ to be the set of the runs of $\psyst$ projected onto a single process. Note that, due to symmetry, we can assume w.l.o.g. that the runs are projected onto process $1$. Formally, $\exec = \{ \proj{\pi}(1) \mid \pi \text{ is a run of } \psyst\}$.
We denote by $\execfin$ (resp. $\execinf$) the finite (infinite) executions in $\exec$.

For specifications $\Pspec$ (e.g., \LTL, NFWs) interpreted over infinite (resp. finite) words over the alphabet $S \times (\Actionprts \cup \{\brd\}) \times S$ of transitions,\footnote{In this way we can also capture atomic propositions on edges or states since these atoms may be pushed into the rendezvous label.}
 the {\em Parameterized Model Checking Problem} (PMCP) for $\Pspec$ is to decide, given a template $\proctemp$, and a specification $\varphi \in \Pspec$, if all executions in $\execinf$ (resp. $\execfin$) satisfy $\varphi$.

\head{Discrete Timed Networks}
We refer the reader to~\cite{ADM04} for a formal definition of timed networks. Here we describe the templates and informally describe the semantics.
Fix a set $C$ of \emph{clocks}.
A {\em timed network template} is a finite LTS $\tpl{Q,I,R,\Actionprts}$.
We associate to each letter $\msg{a}_i \in \Actionprts$ a {\em command} $r(\msg{a}_i) \subseteq C$ and a {\em guard} $p(\msg{a}_i)$.
A guard $p$ is a Boolean combination of predicates of the form $c \bowtie x$ where $c \in \Nat$ is a constant, $x \in \clocks$ is a clock, and ${\bowtie} \in \{<,=\}$.

Intuitively, a discrete timed network consists of the parallel composition of $n \in \Nat$ template processes, each running a copy of the template.
Each copy has a local state $(q,t)$, where $q \in \timedStates$ and $t:\clocks \to \Nat$.
A rendezvous action $\msg{a}$ is \emph{enabled} if there are $k$ processes in local states $(q_i,t_i)$ ($i \in [k]$) and there are edges $(q_i,\msg{a}_i,q'_i) \in R$ such that the clocks $t_i$ satisfy the guards $p(\msg{a}_i)$.
The rendezvous action is \emph{taken} means that the $k$ processes change state (to $q'_i$) and each of the clocks in $r(\msg{a}_i)$ is reset to $0$.
The network evolves non-deterministically, in steps: either all clocks advance by one time unit (so every $t(c)$ increases by one)\footnote{Alternatively, as in \cite{ADM04}, one can let time advance by any amount.} or a rendezvous action $\msg{a} \in \Actionprts$ is taken.
For a timed network template $T$ let $\psysttimed^n$ denote the timed network composed of $n \in \Nat$ templates $T$ and let $\psysttimed$ denote the union of the networks $\psysttimed^n$ for $n \in \Nat$.

Given a timed network template $T$ one can build an equivalent RB-template $P$, i.e., $\exec = \execLimited$.
The key insight is that the passage of time, that causes all clocks to advance by one time unit, is simulated by symmetric broadcast, and timed-guards are pushed into the template states.
The RB-system $\psyst$ requires only a finite number of states since clock values bigger than the greatest constant appearing on the guards are collapsed to a single abstract value (cf.~\cite{ADM04}).

\head{Useful lemmas}
We state a few simple but useful lemmas. 
The first ``RB-System Composition'' lemma states that, by partitioning processes of an RB-system into independent groups, a system with many processes can simulate in a single run multiple runs of smaller systems. If the simulated paths contain no broadcasts then the transitions of the simulated paths can be interleaved in any order. Otherwise, all simulated runs must have the same number of broadcasts, and the simulations of all the edges before the $i$'th broadcast on each simulated path must complete before taking the $i$'th broadcast on the simulating combined path.

\begin{lemma}
\label{lem: rb-system composition}
  A system $\sysinst{n}$ can, using a single run, partition its processes into groups each simulating a run of a smaller system. All simulated paths must have the same number of broadcasts.
\end{lemma}

Consider now an RB-system $\sysinst{n}$, and two configurations $f,f'$ in it such that the number of processes in each state in $f$ is equal to that in $f'$, i.e., such that $|f^{-1}(s)| = |f^{\prime -1}(s)|$ for every $s \in S$. We call $f,f'$ {\em twins}.
A finite path $\pi$ of length $m$ for which $\src(\pi_1)$ and $\dst(\pi_m)$ are twins is called a {\em pseudo-cycle}.
For example, for $\proctemp$ in Figure~\ref{fig: process_example}, the following path in $\sysinst{4}$ is a pseudo-cycle that is not a cycle: $(p,q,q,r) \xrightarrow{((3,\msg{c}_1), (4,\msg{c}_2))} (p,q,r,p) \xrightarrow{((2,\msg{c}_1), (3,\msg{c}_2))} (p,r,p,p) \xrightarrow{((3,\msg{a}_1), (4,\msg{a}_2))} (p,r,q,q)$.

\begin{figure}[!htb]
           \begin{floatrow}
             \ffigbox{
\begin{tikzpicture}[->,>=latex,node distance=0.4cm,bend angle=25,auto]

\tikzset{every state/.style={circle,minimum size=.4cm,inner sep=0cm}}
\tikzset{every edge/.append style={font=\small}}

\node[initial, state] (t1) {$p$};
\node[state] (t2) [right= of t1] {$q$};
\node[initial, state] (t3) [below= of t1] {$r$};

\path (t1) edge [bend left] node {{$\msg{a}_1$}} (t2);
\path (t1) edge [bend right, below] node {$\msg{a}_2$} (t2);
\path (t3) edge [bend left, left] node {$\msg{c}_2$} (t1);
\path (t2) edge [bend left] node {$\msg{c}_1$} (t3);

\end{tikzpicture}

%
%
%
%
%
%
}{\caption{R-template with $k=2$.}\label{fig: process_example}}
             \ffigbox{
\begin{tikzpicture}[->,>=latex,node distance=0.6cm,bend angle=25,auto]

\tikzset{every state/.style={circle,minimum size=.4cm,inner sep=0cm}}
\tikzset{every edge/.append style={font=\small}}

\node[initial,state] (t1) {$p$};
\node[state] (t2) [right= of t1] {$q$};

\path (t1) edge [loop above] node {{$\msg{a}_1$}} (t1);
\path (t1) edge [loop below] node {{$\brd$}} (t1);
\path (t1) edge [bend left] node {{$\msg{a}_2$}} (t2);
\path (t2) edge [bend left, below] node {$\brd$} (t1);

\end{tikzpicture}

%
%
%
%
%
%
}{\caption{RB-template with $k=2$.}\label{fig: process not regular}}
           \end{floatrow}
             \ffigbox{\includegraphics[scale = 0.4]{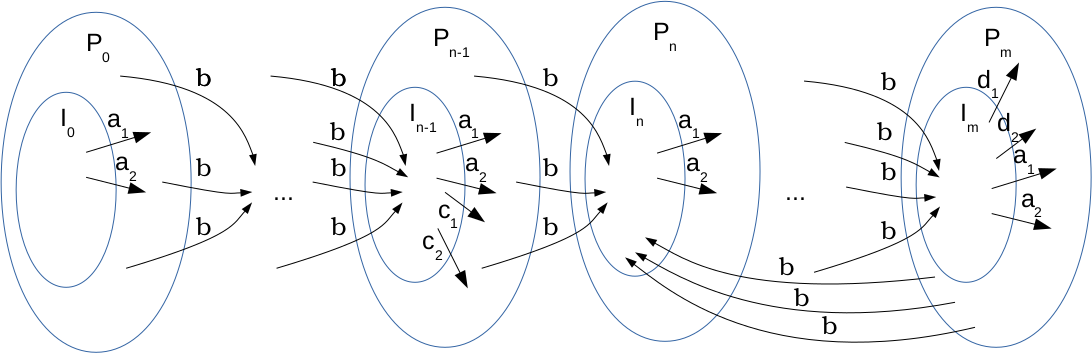}}{\caption{A high level view of the reachability-unwinding lasso.}\label{fig:lasso}}
           
\end{figure}

\begin{lemma}\label{lem: psuedo cycle pumping}
    By renaming processes after each iteration, a pseudo-cycle $\pi$ can be pumped to an infinite path which repeatedly goes through the actions on $\pi$.
\end{lemma}

\section{The Reachability-Unwinding of a Process Template}\label{sec: unwinding}
Given template $\proctemp = \tup{S,I,R,\Actionprts \cup \{\brd\}}$, our goal in this section is to construct a new process template $\procuwd = \tup{S^\uwd,I^\uwd,R^\uwd,\Actionprts \cup \{\brd\}}$, called the {\em reachability-unwinding}
of $\proctemp$, see Figure~\ref{fig:lasso}.
The template $\procuwd$ will play a role in all our algorithms for solving the PMCP of RB-systems. Intuitively, $\procuwd$ is obtained by alternating the following two operations: \emph{(i)} taking a copy of $\proctemp$ and removing from it all unreachable rendezvous edges; and \emph{(ii)} unwinding on broadcast edges.
This is repeated until a copy is created which is equal to a previous one, we then stop and close the unwinding back into the old copy, forming a high-level lasso structure.

Technically, it is more convenient to first calculate all the desired copies and then to arrange them in the lasso. Thus, we first calculate, for $0 \leq i \leq m$ (for an appropriate $m$), an R-template $\proctemp_i = \tup{S_i,I_i,R_i,\Actionprts}$ which is a copy of $\proctemp$ with initial states redesignated and all broadcast edges, plus some rendezvous edges, removed.
Second, we take $\proctemp_0, \dots, \proctemp_m$ and combine them, to create the single process template $\procuwd$, by connecting the states in $\proctemp_i$ with the initial states of $\proctemp_{i+1}$ ($\proctemp_n$ for $i=m$, where $n \le m$ is determined by the lasso structure) with broadcast edges, as naturally induced by $\proctemp$.

Construct the R-template $\proctemp_i = \tup{S_i,I_i,R_i,\Actionprts}$ (called the $i$'th {\em component} of $\procuwd$) recursively: for $i=0$, we let $I_0 := I$; and for $i > 0$ we let $I_i := \{ s \in S \mid (h,\brd,s) \in R \text{ for some } h \in S_{i-1} \}$ be the set of states reachable from $S_{i-1}$ by a broadcast edge. The elements $S_i$ and $R_i$ are obtained using the following {\em saturation} algorithm, which is essentially a breadth-first search: start with $S_i := I_i$ and $R_i := \emptyset$; at each round of the algorithm, consider in turn each edge $e=(s, \msg{a}_h, t) \in R \setminus R_i$; if for every $l \in [k]$ there is some edge $(s', \msg{a}_l, t') \in R$ with $s' \in S_i$, then add $e$ to $R_i$ and add $t$ (if not already there) to $S_i$. The algorithm ends when a fixed-point is reached. Observe a property of this algorithm: if $(s,\msg{a}_h,t) \in R_i$ then for all $l \in [k] \setminus \{h\}$ there exists $s',t' \in S_i$ such that $(s',\msg{a}_l,t') \in R_i$.

Now, $\proctemp_i$ is completely determined by $I_i$ (and $\proctemp$), and so there are at most $2^{|S|}$ possible values for it. Hence, for some $n \leq m < 2^{|S|}$ it must be that $\proctemp_n = \proctemp_{m+1}$. We stop calculating $\proctemp_i$'s when this happens since for every $i \in \NatZero$ it must be that $\proctemp_i = \proctemp_{n + ((i-n)\mod r)}$, where $r=m+1-n$. We call $n$ the {\em prefix length} of $\procuwd$ (usually denoted by $\psi$), call $r$ the {\em period} of $\procuwd$, and for $i \in \NatZero$, call $n + ((i-n) \mod r)$ the {\em associated component number} of $i$, and denote it by $\comp(i)$.

We now construct from $\proctemp_0, \dots, \proctemp_m$ the template $\procuwd = \tup{S^\uwd,I^\uwd,R^\uwd,\Actionprts \cup \{\brd\}}$, as follows: \emph{(i)} $S^\uwd := \cup_{i=0}^m (S_i \times \{i\})$; \emph{(ii)} $I^\uwd := I_0 \times \{0\}$  (recall that we also have $I_0 = I$); \emph{(iii)} $R^\uwd$ contains the following transitions: the rendezvous transitions $\cup_{i=0}^m \{ ((s,i), \varsigma, (t,i)) \mid  (s, \varsigma, t) \in R_i \}$, and the broadcast transitions  $\cup_{i=0}^{m-1} \{ ((s,i), \brd, (t,i+1)) \mid  (s, \brd, t) \in R \text{ and } s \in S_i  \}$
and $\{ ((s,m), \brd, (t,n)) \mid  (s, \brd, t) \in R \text{ and } s \in S_m \}$.

We will abuse notation, and talk about the component $P_i$, referring sometimes to $P_i$ as defined before (i.e., without the annotation with $i$), and sometimes to the part of $\procuwd$ that was obtained by annotating the elements of $P_i$ with $i$.

Observe that, by projecting out the component numbers (we will denote this projecting by superscript $\rwd$) from states in $\procuwd$ (i.e., by replacing $(s,i) \in S^\uwd$ with $s \in S$), states and transitions in $\procuwd$ induce states and transitions in $\proctemp$. 
Similarly, paths and runs in $\psyst^\uwd$ can be turned into paths and runs in $\psyst$.
We claim that also the converse is true, i.e., that by adding component numbers, states and transitions in $\proctemp$ can be lifted to ones in $\procuwd$; and that by adding the correct (i.e., reflecting the number of previous broadcasts) component numbers to the states of the transitions of a run in $\psyst$, it too can be lifted to a run in $\psyst^\uwd$.
However, a path in $\psyst$ that is not a run (i.e., that does not start at an initial configuration), may not always be lifted to a path in $\psyst^\uwd$ due to the removal of unreachable edges in the components making up $\procuwd$.

The next lemma says that we may work with template $\procuwd$ instead of $P$.
\begin{lemma}\label{lem: runs can be lifted to unwinding}
  For every $n \in \Nat$, we have that $\runs(\sysinst{n}) = \{ \rho^\rwd \mid \rho \in \runs((\psyst^\uwd)^n) \}$.
\end{lemma}

The following lemma says, intuitively, that for every component $P_i$ there is a run of $\psyst^\uwd$ that ``loads'' arbitrarily many processes into every state of $P_i$.

\begin{lemma}
\label{lem: loading}
For all $b,n \in \Nat$ there is a finite run $\pi$ of $\psyst^\uwd$ with $b$ broadcasts, s.t., $|f^{-1}(s)| \geq n$ for all states $s$ in the component $\proctemp_{\comp(b)}$, where $f = \dst(\pi)$. \qed
\end{lemma}

The following lemma states that the set of finite executions of the RB-system $\psyst$ is equal  to the set of finite  {\em runs of the process template} $\procuwd$ (modulo component numbers).
This is very convenient since, whereas $\psyst$ is infinite, $\procuwd$ is finite.
Unfortunately, when it comes to infinite executions of $\psyst$ we only get that they are contained in (though in many cases not equal to) the set of infinite runs of $\procuwd$.
This last observation is also true for $\proctemp$: consider for example Figure~\ref{fig: process not regular} without the $\brd$ edges, and an infinite repetition of the self loop.

\begin{lemma}\label{lem: unwinding captures executions}
  $\execfin[\psyst] = \{\pi^\rwd \mid \pi \in \runs(\procuwd), |\pi| \in \Nat \}$; and $\execinf[\psyst] \subseteq \{\pi^\rwd \mid \pi \in \runs(\procuwd), |\pi| = \infty \}$
  \end{lemma}

\head{Solving PMCP for regular specifications}
Given $\proctemp = \tup{S,I,R,\Actionprts \cup \{\brd\}}$, let $\execNFW$ denote the reachability-unwinding $\procuwd$ viewed as an automaton (NFW), with all states being accepting states, and transitions $e$ are labeled $e^\rwd$ (i.e., they have the component number removed). Formally, $\execNFW = \tup{R,S^\uwd,I^\uwd,R^\prime,S^\uwd}$, so the input alphabet of $\execNFW$ is $R$ (the transition relation of $\proctemp$), and $R^\prime := \{ (s, (s^\rwd,\sigma,t^\rwd), t) ~ | ~ (s, \sigma, t) \in R^\uwd \}  \subseteq S^\uwd \times R \times S^\uwd$. Hence:

\begin{theorem}\label{finite RB PMCP is in pspace}
  The PMCP of RB-systems (resp. discrete timed networks) for regular specifications is in \PSPACE\ (resp. \EXPSPACE)
\end{theorem}

\section{Solving PMCP of Liveness Specifications}\label{sec: PMCP infinite}

In this section we show how to solve the PMCP for specifications concerning infinite executions.
We begin with the following lemma showing that, if we want to use the automata theoretic approach, classical automata models (e.g. B\"{u}chi, Parity) are not up to the task.

\begin{lemma}\label{lem: infinite exec no regular}
  There is a process template $\proctemp$ such that $\execinf$ is not $\omega$-regular.
\end{lemma}
\begin{proof}
  Consider the process template given in Figure~\ref{fig: process not regular}. It is not hard to see that in every infinite run of $\sysinst{n}$ there may be at most $n-1$ consecutive rendezvous transitions before a broadcast transition, resetting all processes to state $1$, is taken. Overall, we have that $\execinf$ is the set of words of the form $\msg{a}_1^{n_1} \msg{a}_2^{m_1} \brd \allowbreak \msg{a}_1^{n_2} \msg{a}_2^{m_2} \brd \dots$, where $m_i \in \{0,1\}$ for every $i$, and $\limsup n_i < \infty$. This language is not $\omega$-regular since the intersection of its complement with $\{\msg{a}_1, \brd\}^\omega$ is not $\omega$-regular (because it contains no ultimately periodic words).
\qed \end{proof}

In light of Lemma~\ref{lem: infinite exec no regular}, we turn our attention to a stronger model, called BSW \cite{b10}.
Thus, we solve the PMCP for liveness specifications as follows: given a process template $\proctemp$, we show how to build a BSW $\execBSW$ accepting exactly the executions in $\execinf$.
Model checking of a specification given by a BSW $\A'$ accepting all undesired (i.e., bad) executions, is thus reduced to checking for the emptiness of the intersection of $\execBSW$ and $\A'$.

\head{Defining the Automaton $\execBSW$}
We now describe the structure of the BSW $\execBSW$ (in fact we define a BW) accepting exactly the executions in $\execinf$.

An important element in the construction is a classification of the edges in $\procuwd$ into four types: \bgood, \good, \sgood, and \bad. The \bad\ edges are those that appear at most finitely many times on any execution in $\execinf$. An edge is \bgood\ if it appears infinitely many times on some execution in $\execinf$ with finitely many broadcasts, but only finitely many times on every execution which has infinitely many broadcasts. An edge $e$ is \good\ if there is some run $\pi \in \execinf$ with infinitely many broadcasts on which $e$ appears unboundedly many times between broadcasts, i.e., if for every $n \in \Nat$ there are $i<j \in\Nat$ such that $\pi_i \dots \pi_j$ contains $n$ occurrences of $e$ and no broadcast edges. An edge which is neither \bgood, \good, nor \bad\ is \sgood. By definition,  \bgood\ and \good\ edges are not broadcast edges. Since the set $\execinf$ is infinite, it is not at all clear that the problem of determining the type of an edge is decidable. Indeed, this turns out to be a complicated question, and we dedicate Section~\ref{sec: deciding types of edges} to show that one can decide the type of an edge.

The automaton $\execBSW$ is made up of three copies of $\execNFW$ (called ${\execBSW}^{1}$, ${\execBSW}^2$, ${\execBSW}^3$), as follows: ${\execBSW}^1$ is an exact copy of $\execNFW$; the copy ${\execBSW}^2$ has only the \good\ and \sgood\ edges left; and ${\execBSW}^3$ has only the \bgood\ and \good\ edges left (and in particular has no broadcast edges). Furthermore, for every edge $(s, \sigma, s')$ in ${\execBSW}^{1}$ we add two new edges, both with the same source as the original edge, but one going to the copy of $s'$ in the copy ${\execBSW}^2$, and one to the copy of $s'$ in the copy ${\execBSW}^3$. The initial states of $\execBSW$ are the initial states of ${\execBSW}^1$. For the acceptance condition: every state in ${\execBSW}^2$ and ${\execBSW}^3$ is a B\"uchi-state, and there is a single counter $C \in \Gamma_B$ that is incremented whenever an orange rendezvous edge is taken in ${\execBSW}^2$ and reset if a broadcast edge is taken in ${\execBSW}^2$.

Formally, given a process template $P = \tpl{ S,I, R, \Actionprts \cup \{ \brd \} }$ and its unwinding $P^\uwd = \tpl{ S^\uwd, I^\uwd, R^\uwd, \Actionprts \cup \{ \brd \} }$  define $\BSWAPInf = \tpl{ \Sigma, Q,Q_0, \Gamma,\delta,\Phi }$ as:
\begin{itemize}
\item The input alphabet $\Sigma$ is the edge relation $R$ of template $P$.
\item The state set $Q$ is $\{ (i,s) \st s \in S^\uwd, i \in \{ 1,2,3 \} \}$.
\item The initial state set $Q_0$ is $\{ (1,s) \st  s \in I^\uwd \}$.
\item There is one counter, $\Gamma = \{c\}$.
\item The transition relation $\delta$ is $\delta_1 \cup \delta_2 \cup \delta_3$, where: $\delta_1$ consists of all tuples 
$((1,s_1),(s_1^\rwd,\sigma,s_2^\rwd), \epsilon, (i,s_2))$ such that $(s_1,\sigma,s_2) \in R^\uwd, i \in \{1,2,3\}$; and
$\delta_3$ consists of all tuples $\{ ((3,s_1),(s_1^\rwd,\sigma,s_2^\rwd), \epsilon, (3,s_2))$ such that 
$ (s_1,\sigma,s_2) \in R^\uwd$ is $\bgood$ or $\good$; and 
$\delta_2$ consists of all tuples $((2,s_1),(s_1^\rwd,\sigma,s_2^\rwd), upd^{\sigma,\rho} , (2,s_2))$ such that 
$\rho := (s_1,\sigma,s_2) \in R^\uwd$ is $\good$ or $\sgood$, 
and $upd^{\sigma,\rho}$ is the single operation $c := 0$ if $\rho$ is $\sgood$ and $\actn(\sigma) = \brd$,  and $upd^{\sigma,\rho}$ is the single operation $c := c + 1$ if $\rho$ is $\sgood$ and $\actn(\sigma) \neq \brd$, and $upd^{\sigma,\rho}$ is the empty sequence $\epsilon$ if $\rho$ is $\good$. Here $\epsilon$ is the empty sequence of operations (i.e., do nothing to the counter).
\item The acceptance condition $\Phi$ states that $\limsup_i c(\rho,i) < \infty$ (i.e., counter $c$ must be bounded) and some state $q \in Q \setminus \{(1,s) \st s \in S^\uwd\}$ is visited infinitely often.
\end{itemize}

\begin{lemma}\label{lem: edge types and pseudo cycles}
  An edge $(s_1,\sigma,s_2)$ of $\procuwd$ is: \emph{(i)} \bad\ iff it does not appear on any pseudo-cycle of $\psyst^\uwd$;
  \emph{(ii)} \bgood\ iff it appears on a pseudo-cycle of $\psyst^\uwd$ with no broadcasts, but not on any that contain broadcasts;
  \emph{(iii)} \good\ iff it appears on a pseudo-cycle of $\psyst^\uwd$ with no broadcasts, that is part of a bigger pseudo-cycle with broadcasts;
  \emph{(iv)} \sgood\ iff it appears on a pseudo-cycle $C$ of $\psyst^\uwd$ that has broadcasts, but not on any without broadcasts.
\end{lemma}

The following lemma states that we can assume that pseudo-cycles mentioned in Lemma~\ref{lem: edge types and pseudo cycles} (that have broadcasts) are of a specific form.

\begin{lemma}\label{lem: pseudo cycle implies one with r broadcasts}
  An edge $e$ appears on a pseudo-cycle $D$ in $\psyst^\uwd$, which contains broadcasts, iff it appears on a pseudo cycle $C$ of $\psyst^\uwd$ containing exactly $r$ broadcast transitions and with all processes starting in the component $\proctemp_n$, where $n,r$ are the prefix length and period of $\procuwd$, respectively. Furthermore, $C$ preserves any nested pseudo-cycles of $D$ that contain no broadcasts.
\end{lemma}

\begin{theorem}\label{thm: BSW correctness}
 The language recognized by $\execBSW$ is exactly $\execinf$.
\end{theorem}
\begin{proof}[sketch]
The fact that every word in $\execinf$ is accepted by $\execBSW$ follows in a straightforward way from its construction.
For the reverse direction, given $\alpha \in \execinf$ with an accepting run $\Omega$ in $\execBSW$, we need to construct a run $\pi$ in $\psyst$ whose projection on process $1$ is $\alpha$. We consider the interesting case that $\alpha$ has infinitely many broadcasts (and thus finitely many \bad\ and \bgood\ edges). The challenging part is how to make process $1$ trace the suffix $\beta$ of $\alpha$ containing only \good\ and \sgood\ edges. Since $\Omega$ is accepting, counter $C_2$ is bounded on $\Omega$. Hence, there is a bound $\mathfrak{m}$ on the number of \sgood\ edges in $\beta$ between any $r$ broadcasts, where $r$ is the period of $\procuwd$.

For every \good\ (resp. \sgood) edge $e$ of $\proctemp$ that appears on $\beta$, by Lemmas~\ref{lem: edge types and pseudo cycles}, \ref{lem: pseudo cycle implies one with r broadcasts}, there is a pseudo-cycle $C_e$ with $r$ broadcasts on which $e$ appears. Furthermore, if $e$ is \good\ it actually appears on an inner pseudo-cycle of $C_e$ without broadcasts. Let $E_\good$ (resp. $E_\sgood$) be the set of $\good$ (resp. $\sgood$) edges that appear infinitely often on $\alpha$. By taking exactly enough processes to assign them to one copy of $C_e$ for every $e \in E_\good$, and $\mathfrak{m}$ copies of $C_e$ for every $e \in E_\sgood$, and composing them using Lemma~\ref{lem: rb-system composition} we can simulate all these copies of these pseudo-cycles in one pseudo-cycle $D$ also with $r$ broadcasts. By Lemma~\ref{lem: psuedo cycle pumping}, we can pump this pseudo-cycle forever. Furthermore, between broadcasts we have freedom on how to interleave the simulations. We make process $1$ trace $\beta$ by making it successively swap places with the right process in the group simulating a copy of the cycle $C_e$ where $e$ is the next edge on $\beta$ to be traced (just when the group is ready to use that edge). The key observation is that once a group is used by process $1$ there are two options. If it is a group corresponding to a \good\ edge then we can make the group (after $1$ leaves it) traverse the inner pseudo-cycle (the one without broadcasts) thus making it ready to serve process $1$ again. If the group corresponds to an \sgood\ edge $e$, then it will only be reusable when the whole pseudo-cycle $C_e$ completes (since there is no inner pseudo-cycle to use), i.e., after $r$ broadcasts. However, since there are $\mathfrak{m}$ groups for each such edge, and $\mathfrak{m}$ bounds from above the number of \sgood\ edges that need to be taken by process $1$ between $r$ broadcasts.
\qed
\end{proof}

As we show (Section~\ref{sec: deciding types of edges}, Theorem~\ref{thm: edge type decidability}), the problem of determining the type (\bgood, \good, \sgood, or \bad) of an edge in $\procuwd$ is decidable, hence, we conclude this section by stating our main theorem (the proof is now immediate).

\begin{theorem} \label{thm: main dec}
  The PMCP (of RB-systems or discrete timed networks) for BW- or SW-specifications or complements of specifications given by BSW, is decidable.
\end{theorem}

\subsection{Deciding Edge Types}\label{sec: deciding types of edges}

\begin{theorem}\label{thm: edge type decidability}
  Given a process template $\procuwd$, the problem of determining the type (\bgood, \good, \sgood, \bad) of an edge $e$ in $\procuwd$ is decidable.
\end{theorem}

A key observation for proving Theorem~\ref{thm: edge type decidability} is that
by Lemma~\ref{lem: edge types and pseudo cycles}, the type of an edge can be decided by looking for witnessing pseudo-cycles $C$ in $\psyst^\uwd$. Indeed, a witness can determine if an edge is \good\ or not. If not, another witness can determine if it is \sgood\ or not, and the last witness can separate the \bgood\ from the \bad.
We will show an algorithm that given an edge that is not \good\ tells us if it is \sgood\ or not. The algorithm can be modified to check for the other types of witnesses without much difficulty.

By Lemma~\ref{lem: pseudo cycle implies one with r broadcasts}, we can assume that the pseudo-cycle $C$ we are looking for has very specific structure. Our algorithm uses linear programming, in a novel and interesting way, to detect the existence of such a pseudo-cycle $C$.

\head{Counter Representation}
Given a process template $\proctemp = \tup{S,I,R,\Actionprts \cup \{\brd\}}$, let $d = |S|$, and fix once and for all some ordering $s_1, s_2, \dots, s_d$ of the states in $S$. We associate with every configuration $f$ in $\psyst$ a vector $f^\sharp := (|f^{-1}(s_1)|, \dots, \allowbreak |f^{-1}(s_d)|) \in \NatZero^d$, called the {\em counter representation} of $f$. 
We also associate with every transition $t = (f, \sigma, g)$ the vector $t^\sharp := g^\sharp - f^\sharp$ representing the change in the number of processes in each state. If $t$ is a rendezvous transition then $g^\sharp - f^\sharp$ is completely determined by the action $\msg{a} \in \Actions$ taken in $\sigma$. Indeed, if $\sigma = ((i_1, \msg{a}_1), \dots, (i_k, \msg{a}_k))$ then $g^\sharp - f^\sharp = \msg{a}^\sharp$, where $\msg{a}^\sharp \in \NatZero^d$ is the vector defined by letting $\msg{a}^\sharp(s) := |\{ j \in [k] \mid \dst(\edge(a_j)) = s \}| - |\{ j \in [k] \mid \src(\edge(a_j)) = s \}|$ for every $s \in S$.

Given $u \in \Rat^d$, and a sequence of vectors $\varrho = \varrho_1 \dots \varrho_m$ in $\Rat^d$, the pair $\rho = (u, \varrho)$ is called a path from $u$ to $v = u + \Sigma_{i=1}^m \varrho_i$.
We write $\rho_j$ for the vector $u + \Sigma_{i=1}^j \varrho_i$, for every $0 \leq j \leq m$.
The path $\rho$ is {\em legal} if $\rho_j \in \RatGEZ^d$ for every $0 \leq j \leq m$, i.e., if no coordinate goes negative at any point. Given a finite path $\pi_1 \dots \pi_m$ in $\psyst$, we call the path $\pi^\sharp := (\src(\pi_1)^\sharp, \pi_1^\sharp \dots \pi_m^\sharp)$ in $\Rat^d$ its {\em counter representation}. Observe that $\pi^\sharp$ is always a legal path.

\head{Rational Relaxation of VASs}
Vector Addition Systems (VASs) or equivalently Petri nets are one of the most popular  formal  methods  for  the  representation  and  the  analysis  of  parallel
processes \cite{journals/eatcs/EsparzaN94}.
Unfortunately, RB-systems \textbf{cannot} be modelled by VASs since a transition in a VAS only moves a constant number of processes, whereas a broadcast in an RB-system may move any number of processes. 
On the other hand, R-System can be modelled by VASs, and we do use this fact to analyze the behaviour of the counter representation between broadcasts.
Moreover, we note that integer linear programming is a natural fit for describing paths and cycles in the counter representation.
However, in order to apply linear programming to RB-systems we have to overcome two intertwined obstacles:
\emph{(i)} not every path in the counter representation induces a path in $\psyst$, and
\emph{(ii)} since we have no bound on the length of the pseudo-cycle $C$ we cannot have variables describing each configuration on it, and we need to aggregate information. These obstacles are aggravated by the presence of broadcasts.
Another difficulty of applying linear programming to RB-systems arises from the fact that the question of reachability in an RB-system with two (symmetric) broadcast actions and a controller is undecidable (which can be obtained by modifying a result in~\cite{efm99} concerning asymmetric broadcast).

The solution we propose to this problem, which we found to be surprisingly powerful, is to use linear programming but look for a solution in {\em rational} numbers and not in integers. Thus, we introduce the notion of the \emph{rational relaxation} of a VAS, obtained by allowing any non-negative rational multiple of configurations and transitions of the original VAS.
Since our linear programs use homogeneous systems of equations, multiplying a rational solution by a large enough number would yield another solution in integers.
Thus the scaling property obtained a consequence of rational relaxation precludes the possibility of specifying a single controller!
Thinking of the counter representation as vectors of rational numbers also allows us to use geometric reasoning to solve the two problems (i), (ii) described above.
Essentially, by cutting transitions to smaller pieces (which cannot be done at will to integer vectors) and rearranging the pieces, we can transform a description of a path in an aggregated form, as it comes out of the linear program, into one which is legal and can be turned into a path in $\psyst$. We strongly believe that these techniques can be fruitfully used in other circumstances concerning counter-representations, and similar objects (such as vector addition systems and Petri nets).

Due to lack of space, the description of the linear programs we use, as well as the geometric machinery we develop will be published in an extended version.

\bibliography{SBC}

\end{document}